%% file: main.tex
\documentclass[submission,copyright,creativecommons,noncommercial]{eptcs}
\providecommand{\event}{QPL 2014} 

\usepackage{amsmath} 
\usepackage{amsthm} 
\usepackage{amssymb}	
\usepackage{bold-extra} 

\newenvironment{mitem}
{\begin{itemize}
  \setlength{\itemsep}{1pt}
  \setlength{\parskip}{0pt}
  \setlength{\parsep}{0pt}}
{\end{itemize}}

\newtheorem{thm}{Theorem}
\newtheorem*{thm*}{Theorem}

\newtheorem*{prop*}{Proposition}
\newtheorem{lem}[thm]{Lemma}

\theoremstyle{definition}

\newcommand{\ZX}{\textsc{zx}}

\newcommand{\ZZ}{\mathbb{Z}}

\newcommand{\ket}[1]{\left| #1 \right>} 

\input{tikzit_header.tex}

\title{The \ZX-calculus is complete for the single-qubit Clifford+T group}
\author{Miriam Backens
\institute{Department of Computer Science, University of Oxford \\ Wolfson Building, Parks Road, Oxford, OX1 3QD, UK}
\email{miriam.backens@cs.ox.ac.uk}
}
\def\titlerunning{The \ZX-calculus is complete for the single-qubit Clifford+T group}
\def\authorrunning{Miriam Backens}

\begin{document}

\maketitle

\begin{abstract}
 The \ZX-calculus is a graphical calculus for reasoning about pure state qubit quantum mechanics.
 It is complete for pure qubit stabilizer quantum mechanics, meaning any equality involving only stabilizer operations that can be derived using matrices can also be derived pictorially.
 Stabilizer operations include the unitary Clifford group, as well as preparation of qubits in the state $\ket{0}$, and measurements in the computational basis. 
 For general pure state qubit quantum mechanics, the \ZX-calculus is incomplete: there exist equalities involving non-stabilizer unitary operations on single qubits which cannot be derived from the current rule set for the \ZX-calculus.
 Here, we show that the \ZX-calculus for single qubits remains complete upon adding the operator $T=\left(\begin{smallmatrix}1&0\\0&e^{i\pi/4}\end{smallmatrix}\right)$ to the single-qubit stabilizer operations.
 This is particularly interesting as the resulting single-qubit Clifford+T group is approximately universal, i.e.\ any unitary single-qubit operator can be approximated to arbitrary accuracy using only Clifford operators and $T$.
\end{abstract}

\input{introduction.tex}
\input{ZX_calculus.tex}
\input{completeness.tex}
\input{conclusions.tex}

\section*{Acknowledgements}

The author gratefully acknowledges financial support from the EPSRC.

\bibliographystyle{eptcs}
\bibliography{Clifford_plus_T}

\end{document}

%% file: tikzit_header.tex
\usepackage[svgnames]{xcolor} 

\usepackage{tikz}
\usetikzlibrary{shapes.misc} 

\pgfdeclarelayer{edgelayer}
\pgfdeclarelayer{nodelayer}
\pgfsetlayers{edgelayer,nodelayer,main}

\tikzstyle{none}=[inner sep=0pt]
\tikzstyle{rn}=[circle,fill=Red,draw=Black,line width=0.8 pt,minimum size=5pt,inner sep=0pt]
\tikzstyle{gn}=[circle,fill=Lime,draw=Black,line width=0.8 pt,minimum size=5pt,inner sep=0pt]
\tikzstyle{Hadamard}=[rectangle,fill=Yellow,draw=Black,minimum size=8pt,inner sep=0pt,label={center:$\scriptstyle\mathrm{H}$}]
\tikzstyle{gphase}=[rounded rectangle,rounded rectangle arc length=90,fill=Lime!20,inner sep=2pt]
\tikzstyle{rphase}=[rounded rectangle,rounded rectangle arc length=90,fill=Red!20,inner sep=2pt]
\tikzstyle{normalrect}=[rectangle,fill=white,draw=black,minimum height=12pt,minimum width=14pt,inner sep=0pt]

\tikzstyle{bn}=[circle,fill=black,draw=black,inner sep=1pt]
\tikzstyle{bigcloud}=[cloud,fill=white,draw=black]

\newcommand{\phase}[1]{\begin{tikzpicture}[baseline=-.1cm]
	\begin{pgfonlayer}{nodelayer}
		\node [style=none] (0) at (0, 0.2) {};
		\node [style={#1}] (1) at (0, -0) {};
		\node [style=none] (2) at (0, -0.2) {};
	\end{pgfonlayer}
	\begin{pgfonlayer}{edgelayer}
		\draw [] (0.center) to (1);
		\draw [] (1) to (2.center);
	\end{pgfonlayer}
\end{tikzpicture}}

\newcommand{\Hadamard}[0]{\begin{tikzpicture}[baseline=-0.1cm]
	\begin{pgfonlayer}{nodelayer}
		\node [style=Hadamard] (0) at (0, 0) {};
		\node [style=none] (1) at (0, 0.25) {};
		\node [style=none] (2) at (0, -0.25) {};
	\end{pgfonlayer}
	\begin{pgfonlayer}{edgelayer}
		\draw (0.center) to (1);
		\draw (2.center) to (1);
	\end{pgfonlayer}
\end{tikzpicture}}


%% file: introduction.tex
\section{Introduction}\label{s:introduction}

The \ZX-calculus introduced by Coecke and Duncan \cite{coecke_interacting_2008} is a powerful graphical calculus for pure state qubit quantum mechanics.
It combines the advantages of using both dimensions of a sheet of paper, as in quantum circuit notation, with a built-in system of rewrite rules.
These allow computations to be done graphically without the need to re-state or re-derive circuit identities each time.
Like quantum circuit notation, the \ZX-calculus can be used to express any operation in pure state qubit quantum mechanics, i.e.\ it is \emph{universal}.
Furthermore the rewrite rules can easily be shown to hold true when translated into matrix mechanics; therefore any equality derived in the \ZX-calculus can also be derived in matrix mechanics.
This property is called \emph{soundness}.

A more intricate question is that of \emph{completeness}: Can any equality that is true in matrix mechanics also be derived graphically using the given rule set?
As recently shown, the answer is no: there are equalities in pure state qubit quantum mechanics, even when restricted to single-qubit operators, that cannot be derived graphically using the current set of rewrite rules \cite{zamdzhiev_communication_2013}.
Yet when the set of allowed operations is restricted to stabilizer quantum mechanics, the answer is yes, even for multi-qubit states and operations \cite{backens_zx-calculus_2013}.

Qubit stabilizer quantum mechanics is the fragment of quantum mechanics characterised by the fact that all allowed states are eigenstates of tensor products of the Pauli matrices with global phases.
The stabilizer unitaries are those unitary operations that map tensor products of Pauli matrices to tensor products of Pauli matrices, again up to global phase.
They form a group, called the Clifford group, which is generated by the Hadamard operator $H$, the phase operator $S$, and the controlled-NOT operator $\wedge X$:
\begin{equation}\label{eq:Clifford_generators}
 H = \frac{1}{\sqrt{2}}\begin{pmatrix}1&1\\1&-1\end{pmatrix},\quad S = \begin{pmatrix}1&0\\0&i\end{pmatrix},\quad \wedge X = \begin{pmatrix}1&0&0&0\\0&1&0&0\\0&0&0&1\\0&0&1&0\end{pmatrix}.
\end{equation}
While stabilizer quantum mechanics is an important and active area of research, the number of distinct stabilizer operations on a finite number of qubits are finite.
Hence stabilizer quantum mechanics encompasses only a small fragment of all pure quantum operations on qubits.
Yet, the addition of almost any non-stabilizer operation to the Clifford group allows any unitary operation to be approximated to arbitrary accuracy with the given set of operations.

There is thus a large gap between the \ZX-calculus completeness result for stabilizer operations and the incompleteness result for general pure state qubit quantum mechanics.
Here, we make a step towards closing that gap by showing that the \ZX-calculus is complete for the single-qubit Clifford+T group, where we have added the operation
\begin{equation}\label{eq:TH}
 T = \begin{pmatrix}1&0\\0&e^{i\pi/4}\end{pmatrix}
\end{equation}
to the set in \eqref{eq:Clifford_generators}, but removed the two-qubit controlled-NOT operator.
As $S=T^2$, the generating set for the single-qubit Clifford+T group is $\{H,T\}$.
These two operators suffice to approximate any single-qubit unitary to arbitrary precision \cite{boykin_universal_1999}.

The core of the completeness proof is a normal form theorem for single-qubit Clifford+T operators, based on a result for quantum circuits by Matsumoto and Amano \cite{matsumoto_representation_2008}.
We show how an arbitrary single-qubit Clifford+T operator in the \ZX-calculus can be brought into normal form and also prove that the normal form is unique.
This implies that any equality between single-qubit Clifford+T operators that can be derived in matrix mechanics can also be derived in the \ZX-calculus: as all rewrite rules are invertible, bringing two diagrams into the same normal form directly yields a series of rewrite steps transforming one diagram into the other.

In section \ref{s:preliminaries}, we present the elements and rules of the \ZX-calculus for single-qubit Clifford+T operators, as well as a number of definitions and lemmas used throughout the rest of the paper. Section \ref{s:completeness} contains the completeness proof. Some conclusions and ideas for further work are given in section \ref{s:conclusions}.

%% file: ZX_calculus.tex
\section{Preliminaries}
\label{s:preliminaries}

\subsection{Elements and rules of the \ZX-calculus}\label{s:ZX}

In this paper, we are only considering a small fragment of the \ZX-calculus and only introduce the components and rules needed to show the result.
For a treatment of the full \ZX-calculus see \cite{coecke_interacting_2011}.

Diagrams of the \ZX-calculus consist of nodes and wires between them.
Wires may also end at ``empty nodes'', representing inputs and outputs of the diagram.
Here, we are interested only in line graphs, i.e.\ diagrams in which there are exactly two wires connected to any non-empty node.
The three types of nodes are interpreted as follows:
\[
 \phase{gn,label={[gphase]right:$\phi$}} ::\begin{cases}\ket{0}\mapsto\ket{0}\\\ket{1}\mapsto e^{i\phi}\ket{1}\end{cases}\quad\quad \phase{rn,label={[rphase]right:$\theta$}} ::\begin{cases}\ket{+}\mapsto\ket{+}\\\ket{-}\mapsto e^{i\theta}\ket{-}\end{cases}\quad\quad \Hadamard{} :: \begin{cases}\ket{0}\mapsto\ket{+}\\\ket{1}\mapsto\ket{-}\end{cases}
\]
where the phases $\phi$ and $\theta$ are multiples of $\pi/4$, and $\ket{\pm}=\frac{1}{\sqrt{2}}(\ket{0}\pm\ket{1})$.
A single wire corresponds to an identity operation.
Diagrams are read from bottom to top, i.e.\ the dangling wire at the bottom is considered the input, the one at the top the output.
Connecting the output of one node to the input of another corresponds to serial composition of operators.

For example, \phase{gn,label={[gphase]right:$\pi$}} is the Pauli-$Z$ operator and \phase{rn,label={[rphase]right:$\pi$}} is the Pauli-$X$ operator.
Thus
\begin{center}
 \begin{tikzpicture}
	\begin{pgfonlayer}{nodelayer}
		\node [style=gn,label={[gphase]right:$\pi$}] (0) at (0, -0.25) {};
		\node [style=rn,label={[rphase]right:$\pi$}] (1) at (0, 0.25) {};
		\node [style=none] (2) at (0, 0.5) {};
		\node [style=none] (3) at (0, -0.5) {};
	\end{pgfonlayer}
	\begin{pgfonlayer}{edgelayer}
		\draw (2.center) to (3.center);
	\end{pgfonlayer}
 \end{tikzpicture}
\end{center}
corresponds to applying first $Z$ and then $X$.
The node \phase{gn,label={[gphase]right:$\pi/4$}} represents the $T$-operator defined in \eqref{eq:TH}.

It is easy to see that the three types of nodes with the given phase values are all in the Clifford+T group.
Furthermore, as $\{H,T\}$ is a generating set for this group, any Clifford+T operator can be expressed as a diagram in terms of \Hadamard\ and \phase{gn,label={[gphase]right:$\pi/4$}}.

The \ZX-calulus is not just an alternative notation, it comes with rewriting rules that allow the derivation of equalities between diagrams.
The rules relevant for this paper are the following, where $n$ is an integer:
\begin{center}
 \input{tikz-files/ZX-rules-1.tikz}
\end{center}
\begin{center}
 \input{tikz-files/ZX-rules-2.tikz}
\end{center}
All of these rules also hold with the colours reversed.
Furthermore, the rules are all sound---i.e. they hold true when translated back into Dirac or matrix notation---if equality is taken to be up to global scalar factors.

E.g. by rule (P) with $\phi=\pi$, we have
\begin{center}
 \input{tikz-files/pi-commutation_example.tikz}
\end{center}
In matrix notation, the two diagrams correspond to
\[
 \begin{pmatrix}1&0\\0&-1\end{pmatrix}\begin{pmatrix}0&1\\1&0\end{pmatrix} = \begin{pmatrix}0&1\\-1&0\end{pmatrix}
 \quad\text{and}\quad
 \begin{pmatrix}0&1\\1&0\end{pmatrix}\begin{pmatrix}1&0\\0&-1\end{pmatrix} = \begin{pmatrix}0&-1\\1&0\end{pmatrix}.
\]
The two matrices are not equal, but as
\[
 -\begin{pmatrix}0&1\\-1&0\end{pmatrix} = \begin{pmatrix}0&-1\\1&0\end{pmatrix},
\]
the equality holds up to a scalar factor.

\subsection{Some further definitions and lemmas}

We shall denote the single-qubit Clifford group by $\mathcal{C}_1$; this group contains all diagrams built up from the components given in section \ref{s:ZX} in which the phases of red and green nodes are restricted to integer multiples of $\pi/2$.
It will also be useful to define two further sets of \ZX-calculus operators:
 \begin{center}
  \input{tikz-files/W.tikz}
  $\quad$ and $\quad$
  \input{tikz-files/V.tikz}
 \end{center}
In the remainder of this section we will prove various lemmas about how operators in $\mathcal{C}_1$, $\mathcal{W}$ and $\mathcal{V}$ compose.

\begin{lem}\label{lem:C_normal_form}
 The following two sets each contain a unique representation for each operator $C\in\mathcal{C}_1$:
 \begin{equation}\label{eq:local_Clifford_normal_form}
  \left\{
  \input{tikz-files/C1-1.tikz}
  \right\} \quad\quad\text{and}\quad\quad
  \left\{
  \input{tikz-files/C1-2.tikz}
  \right\},
 \end{equation}
 where in both cases $\alpha,\beta,\gamma\in\{0,\pi/2,\pi,-\pi/2\}$.
\end{lem}
\begin{proof}
 First note that any single-qubit Clifford operator can be written in terms of red and green nodes only, by substituting for the Hadamard nodes using the rule (Eu).
 Then, each such operator must have a representation with no more than three nodes: given any diagram with at least four nodes, either 
 \begin{mitem}
  \item there are two adjacent nodes of the same colour, in which case they can be merged by rule (S), or
  \item there is a node with a phase that is a multiple of $2\pi$, in which case it can be removed by rule (Id), or
  \item there is a node with a phase of $\pi$, in which case it can be moved past a node of the other colour using rule (P) and then merged with another node of the same colour, or
  \item there are three adjacent nodes with phases in the set $\{\pm\pi/2\}$. In this last case, note that
   \begin{center}
    \input{tikz-files/colour-change.tikz}
   \end{center}
   Similar results can be derived for any combination of plus and minus signs in the phases. Hence if there is a sequence of four nodes of alternating colours, all of which have phases in the set $\{\pm\pi/2\}$, we can change the colours of three of them, and thus get two adjacent nodes of the same colour, which can be merged.
 \end{mitem}
 In each of the cases listed above the number of nodes in the diagram can be reduced by applying suitable rewrite rules.
 The strategy works until there are no more than three nodes left.
 Having reduced all diagrams to at most three nodes, it is straightforward---albeit somewhat tedious---to check that the given sets indeed contain a unique representation of each Clifford operator.
\end{proof}

Note that lemma \ref{lem:C_normal_form} shows directly that the \ZX-calculus is complete for single qubit Clifford operators.

\begin{lem}\label{lem:U}
 The following set contains a unique representation of each operator of the form $TC$, where $C\in\mathcal{C}_1$:
 \[
  \mathcal{U} = \left\{
  \input{tikz-files/U.tikz}
  \right\}
 \]
 if $\alpha,\beta,\gamma\in\{0,\pi/2,\pi,-\pi/2\}$.
\end{lem}
\begin{proof}
 This follows immediately from lemma \ref{lem:C_normal_form}.
\end{proof}

\begin{lem}\label{lem:CV}
 Let $C\in\mathcal{C}_1$, $U\in\mathcal{U}$ and $V\in\mathcal{V}$. Then
 \begin{center}
  \input{tikz-files/VC.tikz}
  $\quad$ and $\quad$
  \input{tikz-files/UC.tikz}
 \end{center}
 for some $W\in\mathcal{W}$, $U'\in\mathcal{U}$, $V'\in\mathcal{V}$ and $a,b\in\{0,1\}$. For the particular case of the first equality where $C$ consists solely of $\pi$ phase shifts, $W$ is the identity and we have
 \begin{center}
  \input{tikz-files/V-pi.tikz}
 \end{center}
 with $\bar{V}\in\mathcal{V}\setminus\{V\}$.
\end{lem}
\begin{proof}
 Substitute for $C$ using the first set of normal forms given in lemma \ref{lem:C_normal_form} and for $V$ and $U$ using the definitions of $\mathcal{V}$ and $\mathcal{U}$; the results then follow from straightforward application of the rules of the \ZX-calculus.
\end{proof}

\begin{lem}\label{lem:pi-commutation}
 Suppose $V_1,\ldots,V_n\in\mathcal{V}$ for some positive integer $n$. Then if $a,b\in\{0,1\}$,
 \begin{center}
  \input{tikz-files/VV-pi.tikz}
 \end{center}
 for some $a',b'\in\{0,1\}$ and $V_1',\ldots,V_n'\in\mathcal{V}$.
\end{lem}
\begin{proof}
 By induction on $n$, using the second part of lemma \ref{lem:CV}.
\end{proof}

%% file: completeness.tex
\section{Completeness}\label{s:completeness}

\begin{thm}
 Any single-qubit operator consisting of phase shifts that are multiples of $\pi/4$ and Hadamard operators is either a Clifford operator or it can be written in the normal form
 \begin{equation}\label{eq:UVVW}
  \input{tikz-files/UVVW.tikz}
 \end{equation}
 for some integer $n\geq 0$, where $W\in\mathcal{W}$, $V_1,\ldots,V_n\in\mathcal{V}$ and $U\in\mathcal{U}$.
\end{thm}

\begin{proof}
 It is easy to see that any single-qubit Clifford+T operator can be written solely in terms of \phase{rn,label={[rphase]right:$\pi/2$}} and \phase{gn,label={[gphase]right:$\pi/4$}}.
 To prove the theorem, it thus suffices to show that adding \phase{rn,label={[rphase]right:$\pi/2$}} or \phase{gn,label={[gphase]right:$\pi/4$}} to any Clifford operator or any diagram in normal form yields a diagram that can be rewritten to a Clifford operator or normal form diagram.

 Consider first \phase{rn,label={[rphase]right:$\pi/2$}}.
 This is a Clifford operator, so adding it to a Clifford diagram yields another Clifford diagram.
 Furthermore, 
 \begin{center}
  \input{tikz-files/W-pi_2.tikz}
 \end{center}
 for some $C\in\mathcal{C}_1$, so if $n>0$,
 \begin{center}
  \input{tikz-files/UVVW-pi_2.tikz}
  $=\;$
  \input{tikz-files/UVVC.tikz}
  $=\;$
  \input{tikz-files/U-pi-pi-VVW.tikz}
  $=\;$
  \input{tikz-files/UVVW-prime.tikz}
 \end{center}
 by lemmas \ref{lem:CV} and \ref{lem:pi-commutation}, where $a,b\in\{0,1\}$, $W'\in\mathcal{W}$, $U'\in\mathcal{U}$ and $V_1',\ldots,V_n'\in\mathcal{V}$.
 From lemma \ref{lem:CV}, we also have that, if $n=0$, the diagram resulting from the application of \phase{rn,label={[rphase]right:$\pi/2$}} to a normal form diagram can be rewritten into normal form.
 This covers all the cases.

 Now consider \phase{gn,label={[gphase]right:$\pi/4$}} instead.
 Note that
 \begin{center}
  \input{tikz-files/C-pi_4.tikz}
  $\quad$ and $\quad$
  \input{tikz-files/U-pi_4.tikz}
 \end{center}
 for some $U'\in\mathcal{U}$ and $C'\in\mathcal{C}_1$.
 Furthermore, unless $W$ is the idenitity,
 \begin{center}
  \input{tikz-files/W-pi_4.tikz}
 \end{center}
 for some $V\in\mathcal{V}$.
 Thus adding \phase{gn,label={[gphase]right:$\pi/4$}} to a Clifford operator or a normal form diagram with non-trivial $W$ results in diagrams that can be rewritten to normal form.
 If $W$ is the identity and $n=0$, then the result of adding \phase{gn,label={[gphase]right:$\pi/4$}} will be a Clifford diagram.

 It remains to check what happens when $W$ is the identity and $n>0$.
 For any $V_n\in\mathcal{V}$, we can find $W\in\mathcal{W}$ and $a\in\{0,1\}$ such that
 \begin{center}
  \input{tikz-files/V-pi_4.tikz}
 \end{center}
 Then by lemmas \ref{lem:CV} and \ref{lem:pi-commutation}, the entire diagram can be brought into normal form.

 Thus, whenever \phase{rn,label={[rphase]right:$\pi/2$}} or \phase{gn,label={[gphase]right:$\pi/4$}} is added to a Clifford circuit or normal form diagram, the resulting diagram can be rewritten into a Clifford circuit or normal form diagram, completing the proof.
\end{proof}

\begin{thm}\label{thm:not_identity}
 No normal form diagram as given in \eqref{eq:UVVW} is equal to the identity.
\end{thm}
\begin{proof}
 We will show that in matrix mechanics, no normal form circuit is equal to a scalar multiple of the identity matrix.
 As the \ZX-calculus is sound, this implies that no normal form circuit is equal to the identity within the \ZX-calculus.

 Following \cite{matsumoto_representation_2008}, we use an adaptation of the stabilizer formalism.
 Let $M_{(x,y,z)}:=xX+yY+zZ$, where $X,Y,Z$ are the Pauli matrices.
 We say that a single qubit state $\ket{\psi}$ is \emph{stabilized} by $(x,y,z)$ if $M_{(x,y,z)}\ket{\psi}=\ket{\psi}$.
 It is easy to see that if $(x,y,z)$ stabilizes $\ket{0}$, then $(x,y,z)=(0,0,1)$.

 Let $S$ be the phase gate, and denote the \phase{rn,label={[rphase]right:$\pi/2$}} operator by $R$, so $\mathcal{V}=\{TR, TSR\}$.
 Now suppose $(x,y,z)$ stabilizes some state $\ket{\psi}$.
 Then for any $C\in\mathcal{C}_1$, $C\ket{\psi}$ is stabilized by some expression of the form $(a\sigma(x),b\sigma(y),c\sigma(z))$, where $\sigma$ is some permutation on the set $\{x,y,z\}$ and $a,b,c\in\{\pm1\}$.
 This is because $C\ket{\psi} = (CM_{(x,y,z)}C^{-1})C\ket{\psi}$ and conjugation by a Clifford operator maps the set of Pauli matrices to itself, up to factors of $\pm 1$.
 Furthermore,
 \begin{mitem}
  \item $T\ket{\psi}$ is stabilized by $\frac{1}{\sqrt{2}}(x-y,x+y,z\sqrt{2})$,
  \item $TR\ket{\psi}$ is stabilized by $\frac{1}{\sqrt{2}}(x+z,x-z,y\sqrt{2})$, and
  \item $TSR\ket{\psi}$ is stabilized by $\frac{1}{\sqrt{2}}(z-x,x+z,y\sqrt{2})$.
 \end{mitem}

 We shall consider the effect of applying a normal form diagram to $\ket{0}$.
 First, consider the case where $W$ is the identity and $n=0$, i.e.\ the diagram is simply of the form $TC$ for some Clifford operator $C$.
 Now $TC\ket{0}$ is stabilized by one of the expressions
 \begin{equation}\label{eq:TC_stabilizers}
  \frac{1}{\sqrt{2}}(\pm1,\pm1,0),\quad \frac{1}{\sqrt{2}}(\mp1,\pm1,0), \quad\text{and}\quad (0,0,\pm1).
 \end{equation}
 Even though one of the potential stabilizers is $(0,0,1)$, it is easy to see that $TC$ is not a scalar multiple of the identity for any $C$.

 Next consider the possible stabilizers for $V_1TC\ket{0}$, where $V_1\in\mathcal{V}$.
 These are
 \begin{gather*}
  \frac{1}{2}(\pm1,\pm1,\pm\sqrt{2}),\quad \frac{1}{2}(\mp1,\pm1,\pm\sqrt{2}),\quad \frac{1}{2}(\mp1,\mp1,\pm\sqrt{2}),\quad \frac{1}{2}(\pm1,\mp1,\pm\sqrt{2}),\\
  \frac{1}{\sqrt{2}}(\pm1,\pm1,0),\quad\text{and}\quad \frac{1}{\sqrt{2}}(\mp1,\pm1,0).
 \end{gather*}

 Any stabilizer in the set above can be expressed as
 \begin{equation}\label{eq:x1x2}
  \frac{1}{\sqrt{2^m}}(x_1+x_2\sqrt{2},y_1+y_2\sqrt{2},z_1+z_2\sqrt{2}),
 \end{equation}
 where $m,x_1,x_2,y_1,y_2,z_1,z_2\in\ZZ$ with $m\geq0$.
 Applying a transformation from $\mathcal{V}$ maps that stabilizer to
 \[
  \frac{1}{\sqrt{2^{m+1}}}\left((x_1+z_1)+(x_2+z_2)\sqrt{2},(x_1-z_1)+(x_2-z_2)\sqrt{2},2y_2+y_1\sqrt{2}\right)
 \]
 or
 \[
  \frac{1}{\sqrt{2^{m+1}}}\left((z_1-x_1)+(z_2-x_2)\sqrt{2},(x_1+z_1)+(x_2+z_2)\sqrt{2},2y_2+y_1\sqrt{2}\right).
 \]
 Note that $\mathcal{W}\subset\mathcal{C}_1$, so the effect of $W\in\mathcal{W}$ is at most a permutation of the numbers $x,y,z$ and the introduction of minus signs.
 Thus the stabilizer of $U\ket{\psi}$ for any normal form operator $U$ can be written in the form \eqref{eq:x1x2}.

 Following \cite{matsumoto_representation_2008}, we consider the parity of $x_1,x_2,y_1,y_2,z_1$ and $z_2$ under the transformations given by repeated application of elements of $\mathcal{V}$.
 For the stabilizers given in \eqref{eq:TC_stabilizers}, we have either $x_1$ and $y_1$ odd and the others even, or $z_1$ odd and the others even.
 For a given $a,b$, the parity of $|a-b|$ is the same as that of $a+b$, so the two transformations in $\mathcal{V}$ have the same effects on the parity of $x_1,x_2,y_1,y_2,z_1$ and $z_2$.
 
 If $x_1$ and $y_1$ are odd and the others even, then after application of some $V\in\mathcal{V}$, $x_1,y_1$, and $z_2$ are odd.
 A second application of $V$ leads to a stabilizer where all factors are odd except for $z_1$.
 A third application of $V$ gives a stabilizer where once again $x_1,y_1$, and $z_2$ are odd.
 Thus the parity of these factors changes cyclically.

 If $z_1$ is odd in the beginning and the other factors are even, then after one application of $V$, $x_1,y_1$ and $z_2$ are odd, after which the same cyclical behaviour appears as above.

 Note that if $WV_n\ldots V_1TC$ is to be a scalar multiple of the identity, then $V_n\ldots V_1TC\ket{0}$ must have a stabilizer in the set $\{(0,0,c),(0,c,0)\}$ for some non-zero $c$, i.e.\ either $x_1=x_2=y_1=y_2=0$ or $x_1=x_2=z_1=z_2=0$.
 In particular, $WV_n\ldots V_1TC$ can only be the identity if $V_n\ldots V_1TC\ket{0}$ has a stabilizer in which either $x_1,x_2,y_1,$ and $y_2$ are all even, or $x_1,x_2,z_1,$ and $z_2$ are all even.
 Yet, as shown above, for any $V_n\ldots V_1TC\ket{0}$, the factor $x_1$ in the stabilizer is always odd.
 Thus $WV_n\ldots V_1TC$ is never the identity, completing the proof.
\end{proof}

\begin{lem}\label{lem:adjoint}
 Consider a normal form diagram $D=WV_n\ldots V_1U$.
 Then $D^\dagger$ is equal to some normal form diagram with the same number of copies of elements of $\mathcal{V}$, i.e.\ $D^\dagger=W'V_n'\ldots V_1'U'$ for some $W'\in\mathcal{W}, V_1',\ldots,V_n'\in\mathcal{V}$ and $U'\in\mathcal{U}$.
\end{lem}
\begin{proof}
 By the properties of the dagger functor, $D^\dagger = U^\dagger V_1^\dagger\ldots V_n^\dagger W^\dagger$.
 Now for any $U\in\mathcal{U}$, we can find $C\in\mathcal{C}_1$ such that
 \begin{center}
  \input{tikz-files/U_dagger.tikz}
 \end{center}
 and for any $V\in\mathcal{V}$, we have
 \begin{center}
  \input{tikz-files/V_dagger.tikz}
 \end{center}
 Thus by lemmas \ref{lem:CV} and \ref{lem:pi-commutation},
 \begin{equation}\label{eq:adjoint}
  \input{tikz-files/WVVU_dagger.tikz}
  \;=\;
  \input{tikz-files/WVVU_dagger-1.tikz}
  =\;
  \input{tikz-files/WVVU_dagger-2.tikz}
  =\;
  \input{tikz-files/WVVU_dagger-3.tikz}
  =\;
  \input{tikz-files/WVVU_dagger-4.tikz}
  \;=\;
  \input{tikz-files/WVVU_dagger-5.tikz}
 \end{equation}
 for some $W'\in\mathcal{W}$, $V_0',\ldots,V_n',V_1'',\ldots,V_n''\in\mathcal{V}$, and $U'\in\mathcal{U}$.
 Note that $V_1'',\ldots,V_n''$ is just a relabelling of $V_{n-1}',\ldots,V_0'$.
\end{proof}

\begin{thm}
 The normal form for Clifford+$T$ diagrams given in \eqref{eq:UVVW} is unique.
\end{thm}
\begin{proof}
 Suppose there are two normal form diagrams which are equal but not identical.
 Pick a shortest pair of such diagrams, i.e.\ suppose the topmost nodes in the two digrams have different colours or different phases (or both).
 If the topmost nodes are the same, remove them both and keep going like this until a stage is reached where the remaining topmost nodes are different.
 As the two diagrams are not identical, this must be possible.

 Call these two diagrams $D_1$ and $D_2$.
 As $D_1=D_2$ by assumption, and because any normal form diagram is unitary, it must be the case that $D_1^\dagger\circ D_2$ is equal to the identity.
 We will show that under the given assumptions, $D_1^\dagger\circ D_2$ must be equal to some non-trivial normal form diagram.
 By theorem \ref{thm:not_identity}, this normal form diagram cannot be equal to the identity, thus leading to a contradiction.
 From that we conclude that two normal form diagrams are equal if and only if they are identical.

 Suppose $D_1$ can be written in normal form as $WV_n\ldots V_1U$ and $D_2$ as $W'V_m'\ldots V_1'U'$. 
 The requirement that the topmost nodes of $D_1$ and $D_2$ be different can be satisfied in different ways.
 Where the conditions are not symmetric under interchange of $D_1$ and $D_2$, by lemma \ref{lem:adjoint} it nevertheless suffices to consider just one of the two options.
 We will hence distinguish the following cases:
 \begin{mitem}
  \item $W=W'=1$, $n=m=0$, and the topmost nodes of $U$ and $U'$ differ
  \item $W=W'=1$, $n=0\neq m$, and the topmost nodes of $U$ and $V_m'$ differ
  \item $W=W'=1$, $n,m\neq 0$, and $V_n\neq V_m'$
  \item $W\neq W'$, $n=m=0$
  \item $W\neq W'$, $n=0\neq m$
  \item $W\neq W'$, $n,m\neq 0$
 \end{mitem}

 Firstly, if $W=W'=1$ and $n=m=0$, then $D_1=U$ and $D_2=U'$ with $U,U'\in\mathcal{U}$.
 Now any element of $\mathcal{U}$ can be expressed as $TC$, for some $C\in\mathcal{C}$.
 Thus $D_1=TC$ and $D_2=TC'$, and as $U\neq U'$ we must have $C\neq C'$.
 Therefore,
 \begin{center}
  \input{tikz-files/U-U_dagger.tikz}
 \end{center}

 Secondly, if $W=W'=1$ and $n=0\neq m$, consider $U$ and $V_m'$.
 Note that $U=TC$ for some Clifford operator $C$, and $V_m'=TC'$ for some Clifford operator $C'$.
 Again, the requirement that the topmost nodes of $U$ and $V_m'$ be different means that $C\neq C'$.
 As in the first case, we thus find $U^\dagger V_m' = C''$ for some $C''$.
 Then by lemmas \ref{lem:CV} and \ref{lem:pi-commutation}, $D_1^\dagger\circ D_2$ has a normal form $W''V_{m-1}''\ldots V_1''U''$.
 As $m>0$, this is non-trivial.

 The third case, $W=W'=1$, $n,m\neq 0$, and $V_n\neq V_m'$, can be reduced to a case where $W\neq W'$ by applying \phase{gn,label={[gphase]right:$-\pi/4$}} to both diagrams and using rule (S).

 For $W\neq W'$, we have (after some rewriting),
 \begin{equation}\label{eq:WdaggerW}
  \input{tikz-files/W-W_dagger.tikz}
  \in\left\{
  \input{tikz-files/W-W_dagger-set.tikz}
  \right\}.
 \end{equation}
 Then if $n=m=0$,
 \begin{center}
  \input{tikz-files/UW-WU_dagger.tikz}
 \end{center}
 since
 \[
  \input{tikz-files/W-W_dagger-pi.tikz}
  \in\left\{
  \input{tikz-files/W-W_dagger-pi-set.tikz}
  \right\} = \left\{
  \input{tikz-files/W-W_dagger-pi-set-V.tikz}
  \right\}
 \]
 for some $\alpha,\beta,\gamma\in\{0,\pi/2,\pi,-\pi/2\}$, $a,c\in\{0,1\}$ and $V\in\mathcal{V}$.

 The argument for the case $W\neq W'$ and $n=0\neq m$ is very similar, noting that for any $c\in\{0,1\}$, $\gamma\in\{0,\pi/2,\pi,-\pi/2\}$, and $V\in\mathcal{V}$
 \begin{center}
  \input{tikz-files/V-pi-gamma.tikz}
 \end{center}
 for some $V'\in\mathcal{V}$ and $a,b\in\{0,1\}$.
 Hence by lemmas \ref{lem:CV} and \ref{lem:pi-commutation}, the diagram can be rewritten into normal form.

 Lastly, consider the case where $W\neq W'$ and $n,m\neq 0$.
 By lemma \ref{lem:adjoint}, we can rewrite $D_1^\dagger$ to
 \begin{center}
  \input{tikz-files/W_dagger-pi-VVW.tikz}
 \end{center}
 Now
 \begin{center}
  \input{tikz-files/pi_2-pi-V.tikz}
 \end{center}
 for some $\beta\in\{0,\pi/2,\pi,-\pi/2\}$ and $b\in\{0,1\}$.
 Thus the argument concludes in the same way as in the previous case.

 We have shown that for any pair of normal form diagrams $D_1$ and $D_2$, $D_1^\dagger\circ D_2$ has a non-trivial normal form unless the two diagrams are identical.
 Therefore, by theorem \ref{thm:not_identity} and by unitarity of Clifford+T operators, two normal form diagrams are equal if and only if they are identical, i.e.\ the normal form is unique.
\end{proof}

%% file: conclusions.tex
\section{Conclusions}\label{s:conclusions}

We have shown that the \ZX-calculus is complete for the approximately universal single-qubit Clifford+T group.
The proof yields a unique normal form for single-qubit Clifford+T diagrams.
An obvious next step is to attempt to extend the proof to multiple qubits, possibly by combining the results from this paper with \cite{backens_zx-calculus_2013}.
It would also be useful to implement the normalisation algorithm in the automated graph rewriting system \texttt{Quantomatic} \cite{quantomatic}.

While the Clifford+T group being approximately universal for pure single-qubit quantum mechanics makes this completeness result more interesting, there does not currently exist any notion of approximation of operators within the \ZX-calculus itself. It would be interesting to define an approximate equality relation for diagrams, which could then be used to transform arbitrary line graphs into Clifford+T diagrams, normalise, and compare them.

%% file: main.bbl
\begin{thebibliography}{1}
\providecommand{\bibitemdeclare}[2]{}
\providecommand{\surnamestart}{}
\providecommand{\surnameend}{}
\providecommand{\urlprefix}{Available at }
\providecommand{\url}[1]{\texttt{#1}}
\providecommand{\href}[2]{\texttt{#2}}
\providecommand{\urlalt}[2]{\href{#1}{#2}}
\providecommand{\doi}[1]{doi:\urlalt{http://dx.doi.org/#1}{#1}}
\providecommand{\bibinfo}[2]{#2}

\bibitemdeclare{misc}{quantomatic}
\bibitem{quantomatic}
\emph{\bibinfo{title}{Quantomatic}}.
\newblock
  \bibinfo{howpublished}{\url{https://sites.google.com/site/quantomatic/}}.

\bibitemdeclare{article}{backens_zx-calculus_2013}
\bibitem{backens_zx-calculus_2013}
\bibinfo{author}{Miriam \surnamestart Backens\surnameend}
  (\bibinfo{year}{2014}): \emph{\bibinfo{title}{The {ZX}-calculus is complete
  for stabilizer quantum mechanics}}.
\newblock {\sl \bibinfo{journal}{New Journal of Physics}}
  \bibinfo{volume}{16}(\bibinfo{number}{9}), p. \bibinfo{pages}{093021},
  \doi{10.1088/1367-2630/16/9/093021}.

\bibitemdeclare{inproceedings}{boykin_universal_1999}
\bibitem{boykin_universal_1999}
\bibinfo{author}{P.~Oscar \surnamestart Boykin\surnameend},
  \bibinfo{author}{Tal \surnamestart Mor\surnameend}, \bibinfo{author}{Matthew
  \surnamestart Pulver\surnameend}, \bibinfo{author}{Vwani \surnamestart
  Roychowdhury\surnameend} \& \bibinfo{author}{Farrokh \surnamestart
  Vatan\surnameend} (\bibinfo{year}{1999}): \emph{\bibinfo{title}{On universal
  and fault-tolerant quantum computing: A novel basis and a new constructive
  proof of universality for Shor's basis}}.
\newblock In: {\sl \bibinfo{booktitle}{40th Annual Symposium on Foundations of
  Computer Science (Cat. {No.99CB37039)}}}, \bibinfo{publisher}{{IEEE}}, pp.
  \bibinfo{pages}{486--494}, \doi{10.1109/SFFCS.1999.814621}.

\bibitemdeclare{incollection}{coecke_interacting_2008}
\bibitem{coecke_interacting_2008}
\bibinfo{author}{Bob \surnamestart Coecke\surnameend} \& \bibinfo{author}{Ross
  \surnamestart Duncan\surnameend} (\bibinfo{year}{2008}):
  \emph{\bibinfo{title}{Interacting Quantum Observables}}.
\newblock In: {\sl \bibinfo{booktitle}{Automata, Languages and Programming}},
  \bibinfo{volume}{5126}, \bibinfo{publisher}{Springer Berlin Heidelberg},
  \bibinfo{address}{Berlin, Heidelberg}, pp. \bibinfo{pages}{298--310},
  \doi{10.1007/978-3-540-70583-3\_25}.

\bibitemdeclare{article}{coecke_interacting_2011}
\bibitem{coecke_interacting_2011}
\bibinfo{author}{Bob \surnamestart Coecke\surnameend} \& \bibinfo{author}{Ross
  \surnamestart Duncan\surnameend} (\bibinfo{year}{2011}):
  \emph{\bibinfo{title}{Interacting quantum observables: categorical algebra
  and diagrammatics}}.
\newblock {\sl \bibinfo{journal}{New Journal of Physics}}
  \bibinfo{volume}{13}(\bibinfo{number}{4}), p. \bibinfo{pages}{043016},
  \doi{10.1088/1367-2630/13/4/043016}.

\bibitemdeclare{article}{matsumoto_representation_2008}
\bibitem{matsumoto_representation_2008}
\bibinfo{author}{Ken \surnamestart Matsumoto\surnameend} \&
  \bibinfo{author}{Kazuyuki \surnamestart Amano\surnameend}
  (\bibinfo{year}{2008}): \emph{\bibinfo{title}{Representation of Quantum
  Circuits with Clifford and $\pi/8$ Gates}}.
\newblock {\sl \bibinfo{journal}{{arXiv:0806.3834}}}.
\newblock \urlprefix\url{http://arxiv.org/abs/0806.3834}.

\bibitemdeclare{article}{zamdzhiev_communication_2013}
\bibitem{zamdzhiev_communication_2013}
\bibinfo{author}{Christian \surnamestart Schr\"{o}der~de Witt\surnameend} \&
  \bibinfo{author}{Vladimir \surnamestart Zamdzhiev\surnameend}
  (\bibinfo{year}{2014}): \emph{\bibinfo{title}{The {ZX}-calculus is incomplete
  for quantum mechanics}}.
\newblock {\sl \bibinfo{journal}{{arXiv}:1404.3633}}.
\newblock \urlprefix\url{http://arxiv.org/abs/1404.3633}.

\end{thebibliography}
